\documentclass[11pt]{article}

\usepackage[T1]{fontenc}
\usepackage[utf8]{inputenc}
\usepackage{lmodern}
\usepackage{microtype}
\usepackage[margin=1in]{geometry}
\usepackage{amsmath,amssymb,amsthm}
\usepackage{booktabs}
\usepackage{tabularx}
\usepackage{array}
\usepackage{graphicx}
\usepackage{natbib}
\usepackage[hidelinks]{hyperref}
\usepackage[nameinlink,noabbrev]{cleveref}

\newcommand{\tcu}{temporal--causal unity}

\newcommand{\dd}{\mathrm{d}}
\newcommand{\wrap}{\operatorname{wrap}}
\newtheorem{definition}{Definition}
\newtheorem{proposition}{Proposition}

\newcolumntype{Y}{>{\raggedright\arraybackslash}X}
\newcolumntype{P}[1]{>{\raggedright\arraybackslash}p{#1}}

\title{\textbf{Temporal--Causal Unity as an Operational Framework for Collective Dynamics}\\
\large Causal-Progress Clocks, Synchronization, and Polarization}
\author{%
  Jian Liu$^{1}$ and Dong Sun$^{1}$\\
  \small $^{1}$Xintong Zhima (Beijing) Technology Co., Ltd., Beijing, China\\
  \small \texttt{liujian@bonideas.cn} \quad \texttt{aihubxpro@gmail.com}%
}
\date{July 20, 2026}

\begin{document}
\maketitle

\begin{abstract}
This paper develops temporal--causal unity (TCU), a framework connecting a
process-philosophical thesis---time is the ordered unfolding of causal
change---to an operational model of cognitive and social dynamics. The
framework deliberately separates three claims: an interpretive thesis about
becoming, a measurable causal-progress coordinate, and a stochastic network
model. Causal progress is defined by
$\tau(t)=\int_0^t\lambda(s\mid\mathcal H_s)\,\dd s$, where the nonnegative
event intensity $\lambda$ must be specified independently of the outcome.
Agents carry an orientation phase and an activation amplitude; weighted
interaction, heterogeneous drift, external input, anchoring, and diffusion
govern their evolution in $\tau$. First- and second-harmonic order parameters
separate consensus from bipolar polarization. For the all-to-all noisy
Kuramoto special case with Lorentzian drift width $\Delta$, synchronization
begins at the conditional threshold $K_c=2(\Delta+D)$, not at a universal
constant. Reproducible numerical illustrations \emph{illustrate} (not empirically demonstrate) this threshold,
causal-clock curve collapse, and the consensus--polarization distinction. Six
historical episodes are treated as scope probes rather than validation data.
The paper derives falsifiable hypotheses and an out-of-sample protocol for
comparing causal-progress and chronological-time models. TCU is therefore
offered as a disciplined bridge between process ontology and complex-systems
modeling, not as a replacement for spacetime physics or as an empirically
established identity between time and causation.
\end{abstract}


\section{Introduction}
\label{sec:intro}

Time and causation are usually represented by different objects. A physical
or social process is indexed by a time coordinate, while causal relations are
encoded by laws, interventions, mechanisms, or directed dependencies. This
division is productive, but it leaves a conceptual question: if change is what
gives temporal order empirical content, can temporal passage and causal
unfolding be treated as two descriptions of a single process? Process
philosophy and phenomenology motivate such a view
\citep{heidegger1962,whitehead1978}; work on the thermodynamic arrow emphasizes
the relation between temporal asymmetry and irreversible processes
\citep{reichenbach1956,prigogine1980,price1996}. None of these traditions, by
itself, supplies an operational model for collective cognition.

At the same time, coupled-oscillator, threshold, and opinion-dynamics models
provide precise accounts of coordination, cascades, and collective order
\citep{kuramoto1984,granovetter1978,degroot1974,watts2002,castellano2009}.
Their ordinary time coordinate, however, often hides an important empirical
fact: equal calendar intervals can contain very different amounts of
interaction, exposure, decision, and institutional change. A week with one
salient event and a week with thousands of mutually reinforcing events need
not represent equal progress through a collective transition.

This paper develops \emph{\tcu} (TCU) as a bridge between these levels. Its
central interpretive thesis is:
\begin{quote}
\emph{Time is not an empty container in which causes operate; experienced
temporal order is the ordered unfolding of causal change.}
\end{quote}
The thesis is not presented as a theorem of physics. It becomes scientifically
useful only after an independently measurable progress coordinate and a model
with falsifiable consequences are supplied.

The paper makes four contributions.
\begin{enumerate}
  \item It separates an ontological interpretation, an operational
  causal-progress clock, and a stochastic dynamical model. This separation
  prevents a metaphor from being mistaken for a measurement claim.
  \item It replaces an unconstrained time--cause plane with a monotone event
  coordinate $\tau$ whose intensity must be estimated without using the
  outcome it is meant to explain.
  \item It formulates a phase--activation network model and corrects two common
  category errors: interaction weights are not automatically causal effects,
  and high synchrony is not automatically polarization.
  \item It gives conditional analytical results, reproducible numerical
  illustrations, scope probes, and a protocol by which the framework could
  fail.
\end{enumerate}

The epistemic status is therefore modest but nontrivial. TCU does not derive
the arrow of time, reduce relativistic spacetime, import quantum measurement
into cognition, or validate a social theory by redescribing historical events.
It proposes a precise research program: specify a causal clock before observing
the target trajectory, fit a constrained dynamical model, and test whether that
clock improves prediction and cross-context invariance.

\section{From an Ontological Thesis to an Operational Framework}
\label{sec:foundations}

\subsection{Three levels of claim}

The phrase \emph{time is causation} can express several logically distinct
claims. Conflating them would make the framework immune to evidence. TCU
therefore uses the hierarchy in \cref{tab:levels}.

\begin{table}[t]
\centering
\caption{Three levels of TCU. Only the operational and model levels generate
direct statistical tests.}
\label{tab:levels}
\begin{tabularx}{\textwidth}{P{0.17\textwidth}YY}
\toprule
Level & Claim & Evidential status \\
\midrule
Interpretive &
Temporal passage and directed causal becoming are two aspects of process. &
A philosophical postulate assessed by coherence, scope, and relation to
existing accounts of time.\\
Operational &
The progress of a process can be indexed by a monotone accumulation of
prespecified events or hazards. &
Testable through measurement validity, curve collapse, and out-of-sample
prediction.\\
Dynamical &
Orientations and activations evolve through heterogeneous drift, network
coupling, input, anchoring, and diffusion in causal progress. &
Testable against alternative stochastic network models and null networks.\\
\bottomrule
\end{tabularx}
\end{table}

The interpretive level has three postulates.
\begin{description}
  \item[TCU-1 (duality).] Temporal order is the order in which causal change is
  realized; causal direction is what distinguishes earlier from later within a
  process.
  \item[TCU-2 (trace).] Memory, records, institutions, and material structures
  are present traces of prior process, not literal persistence of the past.
  \item[TCU-3 (slice).] A state at a time is a cross-section of an unfolding
  process. Spatial, informational, cognitive, and institutional descriptions
  are different observables on that cross-section.
\end{description}
These postulates preserve the original intuition while avoiding the stronger
claim that physical space is mathematically eliminable. Relativistic causal
structure remains untouched. Results showing that causal order constrains
substantial spacetime structure under specific assumptions do not establish a
literal identity between time and causation \citep{malament1977}.

\subsection{The causal-progress clock}

Let $t$ denote chronological time and $\mathcal H_t$ the history of observable
events before $t$. Let $\lambda(t\mid\mathcal H_t)\geq 0$ be a prespecified
intensity of process-relevant events.

\begin{definition}[Causal-progress coordinate]
\label{def:clock}
For a process beginning at $t=0$, its causal-progress coordinate is
\begin{equation}
  \tau(t)=\int_0^t \lambda(s\mid\mathcal H_s)\,\dd s .
  \label{eq:clock}
\end{equation}
For discrete observations one may use
$\tau_n=\sum_{m\leq n}w_m$, where $w_m\geq0$ is a preregistered weight for event
$m$.
\end{definition}

This construction resembles compensator or rescaled time in event-process
analysis \citep{brown2002}, but TCU gives it a specific substantive
interpretation: $\tau$ measures progress through a hypothesized causal process.
Candidate events include exposures, replies, protests, policy decisions, or
verified institutional actions. Raw message volume is not necessarily a good
clock; it may be duplicated, automated, or endogenous to the response.

\begin{proposition}[Reparameterization limitation]
\label{prop:reparam}
If $\lambda(t)>0$ is allowed to be chosen freely after the outcome trajectory
is observed, replacing $t$ by $\tau(t)$ has no empirical content.
\end{proposition}

\begin{proof}
Strict positivity makes \cref{eq:clock} monotone and therefore invertible.
Any trajectory $y(t)$ can then be written
$\widetilde y(\tau)=y(t(\tau))$. If $\lambda$ may depend arbitrarily on the
realized $y$, a desired shape for $\widetilde y$ can be manufactured by time
warping. Empirical content requires $\lambda$ to be constrained by independent
measurements, fixed on training data, or preregistered before the target
trajectory is evaluated.
\end{proof}

\Cref{prop:reparam} is a central safeguard. Successful visual alignment after
post-hoc time warping does not support TCU. Support requires improvement on
held-out data or invariance across contexts under one prespecified measurement
rule.

\subsection{Operational translations}

Several evocative terms in the motivating framework can be retained if they
are tied to observables. \Cref{tab:translation} states those translations and
their limits.

\begin{table}[t]
\centering
\caption{From philosophical language to operational constructs.}
\label{tab:translation}
\begin{tabularx}{\textwidth}{P{0.19\textwidth}YY}
\toprule
Motivating term & Operational construct & Required caution\\
\midrule
Causal unfolding & Accumulated, prespecified event intensity $\tau$ &
The event set and weights cannot be selected to fit the outcome.\\
Causal direction & Phase or orientation $\theta_i$ plus a signed response to
an intervention &
Temporal precedence alone does not establish causation.\\
Causal condensation & A measured transition from dispersed to persistent
state concentration &
This is classical state locking, not quantum wave-function collapse.\\
Causal freezing & Long recovery time or strong anchoring after a perturbation &
The term is an analogy, not the quantum Zeno effect.\\
Causal depletion & High alignment together with low response capacity to novel
input &
High consensus alone can be adaptive and does not imply depletion.\\
Social phase transition & A parameter-dependent qualitative change in order
parameters &
Finite social systems may show smooth crossovers rather than thermodynamic
singularities.\\
\bottomrule
\end{tabularx}
\end{table}

\subsection{Causal effects versus influence weights}

Modern causal inference defines effects through interventions, counterfactuals,
or structural assumptions \citep{pearl2009,woodward2003}. A network edge
estimated from communication or similarity is not, by itself, a causal
relation. In what follows, $W_{ij}$ is therefore called an \emph{influence
weight}. It may be interpreted causally only when an identification design
supports that interpretation---for example randomized exposure, a credible
natural experiment, or a fully stated structural causal model. This
distinction prevents the framework's ontological use of \emph{causal} from
licensing unwarranted empirical conclusions.

\section{A Stochastic Phase--Activation Model}
\label{sec:model}

\subsection{State variables and dynamics}

Consider $N$ units, such as individuals, communities, or institutions. Unit
$i$ has an orientation $\theta_i(\tau)\in[-\pi,\pi)$ and nonnegative activation
$a_i(\tau)$. The complex notation
$\psi_i=a_i e^{\mathrm{i}\theta_i}$ is a compact representation; it is not a
quantum state. The proposed dynamics are
\begin{align}
\dd\theta_i
  ={}& \left[
    \nu_i
    + K\sum_{j=1}^{N} W_{ij}a_j
      \sin(\theta_j-\theta_i)
    + \boldsymbol b_i^\top\boldsymbol x(\tau)
    - q_i\sin(\theta_i-\theta_i^{\mathrm{ref}})
  \right]\dd\tau
  +\sqrt{2D_i}\,\dd B_i(\tau),
\label{eq:phase}\\
\frac{\dd a_i}{\dd\tau}
  ={}& a_i\left[\alpha_i+\gamma_i I_i(\tau)-\beta_i a_i\right],
\qquad \beta_i>0 .
\label{eq:amplitude}
\end{align}
Here $W_{ij}\geq0$ is a fixed or slowly varying influence matrix, scaled so
that its spectral radius has a stated value; $K$ is global coupling; $\nu_i$
is intrinsic drift; $\boldsymbol x$ is external input; $q_i$ anchors a unit to
a reference state; $D_i$ is a diffusion coefficient; and $B_i$ are standard
Brownian motions. Equation~\eqref{eq:amplitude} is logistic activation with
input. It avoids interpreting $a_i$ as a probability amplitude.

If the model is simulated in calendar time, the stochastic time-change rule
gives
\begin{equation}
  \dd\theta_i =
  \lambda(t)f_i(\boldsymbol\theta,\boldsymbol a,t)\,\dd t
  +\sqrt{2D_i\lambda(t)}\,\dd \widetilde B_i(t),
  \label{eq:timechange}
\end{equation}
where $f_i$ is the drift in brackets in \cref{eq:phase}. Multiplying the drift
but not the diffusion by $\lambda$ would be an inconsistent change of clock.

\subsection{Dimensions and identifiability}

The phase is dimensionless. If $\tau$ is measured in causal-event units,
$\nu_i$, $K W_{ij}a_j$, $\boldsymbol b_i^\top\boldsymbol x$, and $q_i$ all have
units of inverse causal-event units, while $D_i$ has the same inverse unit.
If $a_i$ is dimensionless, $\alpha_i$, $\gamma_i I_i$, and $\beta_i a_i$ also
have units of inverse causal-event units.

Only products such as $KW_{ij}a_j$ are identified without normalization.
Consequently, $W$, $a$, and $K$ require explicit scale conventions. In
addition, simultaneously estimating a highly flexible $\lambda(t)$ and a
highly flexible dynamical drift is generally nonidentifiable. The empirical
protocol in \cref{sec:empirical} fixes or cross-fits the clock before estimating
the state dynamics.

\subsection{Order parameters: alignment is not polarization}

For nonnegative measurement weights $\omega_i$ with
$\sum_i\omega_i=1$, define circular harmonics
\begin{equation}
  Z_m(\tau)=\sum_{i=1}^{N}\omega_i e^{\mathrm{i}m\theta_i(\tau)},
  \qquad r_m(\tau)=|Z_m(\tau)|,\quad m=1,2.
  \label{eq:harmonics}
\end{equation}
The first harmonic $r_1$ measures one-cluster alignment. The second harmonic
$r_2$ is large for either one cluster or two opposing clusters. A simple
bipolarity diagnostic is
\begin{equation}
  P(\tau)=\max\{0,r_2(\tau)-r_1(\tau)\}.
  \label{eq:polarization}
\end{equation}
Thus a uniform distribution has $r_1\approx r_2\approx0$, consensus has
$r_1\approx r_2\approx1$ and $P\approx0$, while balanced opposition has
$r_1\approx0$, $r_2\approx1$, and $P\approx1$. No threshold such as
$\langle|\theta_i|\rangle>\pi/4$ is invariant to rotation of the angular
origin. More elaborate applications should compare \cref{eq:polarization}
with distributional and group-based measures of polarization
\citep{esteban1994,bramson2017}.

Activation is summarized by
$\bar a=N^{-1}\sum_i a_i$. A depleted system additionally requires a
low response capacity. For a standardized novel perturbation of size
$\varepsilon$ at $\tau_0$, define
\begin{equation}
  C_h=\frac{1}{N\varepsilon}\sum_{i=1}^{N}
  \left|\wrap\!\left(
  \theta_i^{(\varepsilon)}(\tau_0+h)
  -\theta_i^{(0)}(\tau_0+h)\right)\right|.
  \label{eq:capacity}
\end{equation}
High $r_1$ and low $C_h$, not high $r_1$ alone, operationalize rigidity or
depletion.

\subsection{Conditional synchronization threshold}

The full network model has no universal critical coupling. An analytical
benchmark is available under restrictive assumptions: $a_i=1$, $q_i=0$,
all-to-all weights $W_{ij}=1/N$, common diffusion $D$, no external drive, and
intrinsic drifts distributed according to a Lorentzian with half-width
$\Delta$. In the continuum limit, the first-harmonic amplitude obeys
\citep{strogatz2000,acebron2005,ott2008}
\begin{equation}
  \frac{\dd r_1}{\dd\tau}
  =\left[\frac{K}{2}-(\Delta+D)\right]r_1
   -\frac{K}{2}r_1^3 .
  \label{eq:meanfield}
\end{equation}

\begin{proposition}[Conditional onset]
\label{prop:threshold}
Under the assumptions above, the incoherent state loses stability at
\begin{equation}
  K_c=2(\Delta+D),
  \label{eq:kc}
\end{equation}
and for $K>K_c$ the stable nonzero branch is
$r_1^*=\sqrt{1-K_c/K}$.
\end{proposition}

\begin{proof}
Linearizing \cref{eq:meanfield} at $r_1=0$ gives growth rate
$K/2-(\Delta+D)$. It changes sign at \cref{eq:kc}. Setting the right-hand
side of \cref{eq:meanfield} to zero yields the nonzero branch; the cubic
coefficient makes it stable when it exists.
\end{proof}

For a general nonnegative network without row normalization, the leading
spectral mode changes the onset approximately to
$K_c\simeq2(\Delta+D)/\rho(W)$, subject to the mean-field assumptions used in
the network synchronization literature \citep{restrepo2005,arenas2008}.
Degree heterogeneity, correlated drifts, delays, adaptive edges, finite size,
or non-sinusoidal influence can shift or remove this transition. The earlier
formula $2\langle D\rangle/\langle A\rangle$ is therefore not retained.

\subsection{Activation and anchoring}

For constant input $I_i$, \cref{eq:amplitude} has equilibria
\begin{equation}
  a_i^*=0,\qquad
  a_i^*=\frac{\alpha_i+\gamma_i I_i}{\beta_i}
  \quad\text{when }\alpha_i+\gamma_i I_i>0.
  \label{eq:amplitudeeq}
\end{equation}
The positive branch is locally stable. Activation can amplify social
influence through the factor $a_j$ in \cref{eq:phase}, but it is conceptually
distinct from directional alignment.

The anchoring term makes freezing testable without quantum terminology. When
$q_i$ is large relative to bounded countervailing input and coupling, the
state remains near $\theta_i^{\mathrm{ref}}$ and the exit time from that
neighborhood increases. Whether repeated evaluation raises $q_i$ is an
empirical psychological or organizational hypothesis, not a consequence of
quantum measurement.

\section{Numerical Illustrations}
\label{sec:numerics}

\begin{figure}[t]
\centering
\includegraphics[width=\textwidth]{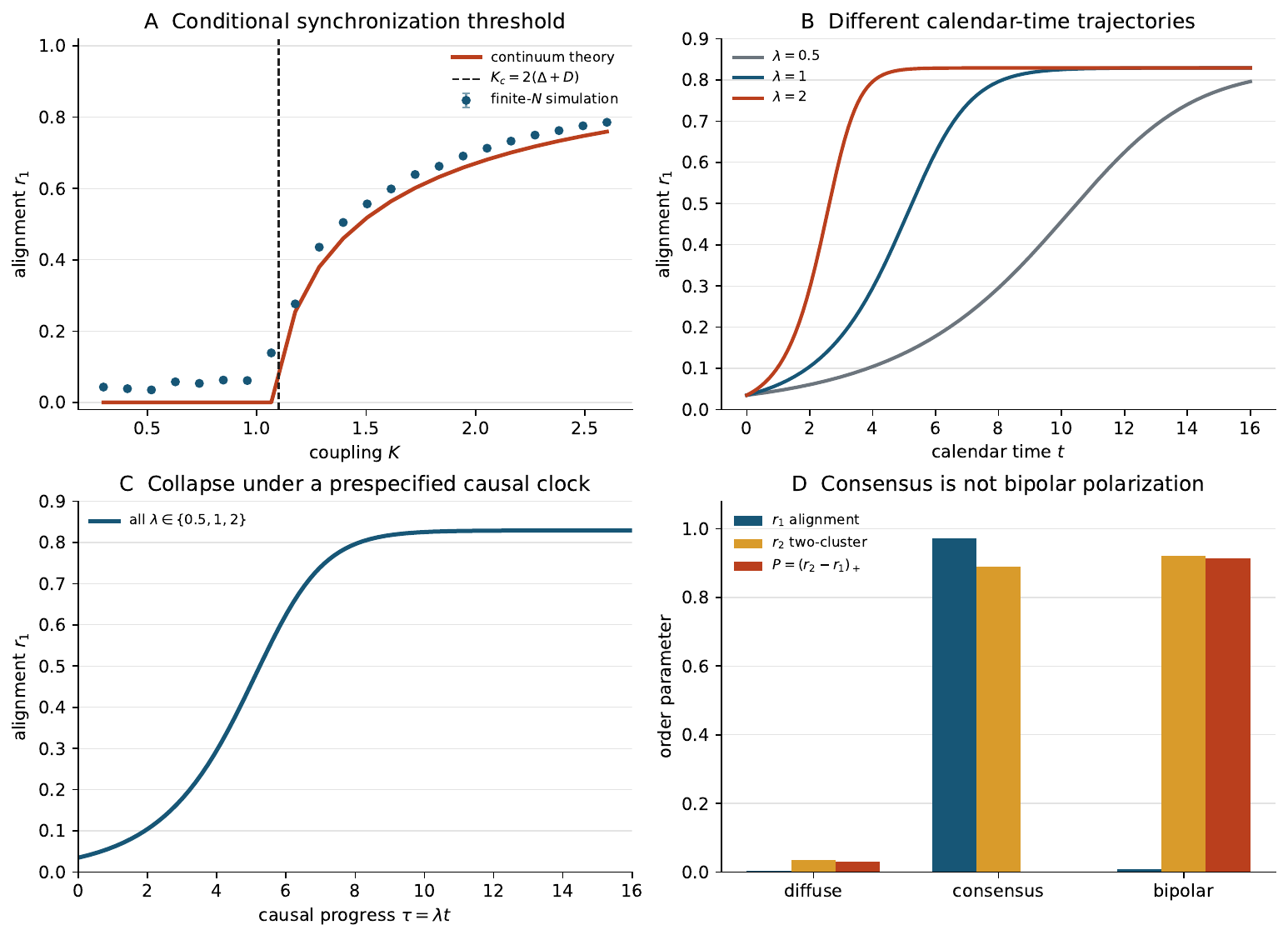}
\caption{Reproducible illustrations, not empirical fits. (A) A finite
$N=600$ Euler--Maruyama simulation of the all-to-all noisy Kuramoto special
case with $\Delta=0.40$, $D=0.15$, step size $0.015$, and fixed seed, compared
with \cref{prop:threshold}. Finite size and truncated Lorentzian tails produce
a small precritical baseline. (B) The same mean-field growth curve appears
at different rates in calendar time when $\lambda\in\{0.5,1,2\}$. (C) The
curves coincide when plotted against the prespecified
$\tau=\lambda t$. (D) First and second harmonics distinguish diffuse,
consensus, and bipolar states. Code is included as an ancillary file.}
\label{fig:numerics}
\end{figure}

\Cref{fig:numerics} checks internal implications of the model. Panel A
recovers the conditional onset predicted by \cref{eq:kc}; it does not show
that any observed society is a Kuramoto system. Panels B and C integrate
\cref{eq:meanfield} in a regime with positive linear growth under three
constant event rates. The collapse in panel C is guaranteed because the
curves were generated in $\tau$. In real data it becomes informative only
if $\lambda$ was specified independently, as required by
\cref{prop:reparam}. Panel D shows why a single synchrony statistic can
misclassify opposition: both consensus and balanced bipolarity have a large
second harmonic, but only consensus has a large first harmonic.

The simulation uses a deterministic grid of Lorentzian quantiles, clips only
the numerically extreme tails, and averages post-burn-in order parameters.
The fixed seed and complete script permit exact regeneration. No historical
case data enter the figure.

\section{Historical Episodes as Scope Probes}
\label{sec:cases}

The motivating framework discussed six global episodes. Redescription of
these episodes cannot validate the model: a flexible vocabulary can be fitted
to almost any sequence after the fact. They are more useful as \emph{scope
probes} that expose measurement choices and possible counterevidence.
\Cref{tab:cases} converts each episode into a prospective design. The cited
studies establish relevant empirical context, not support for TCU.

\begin{table}[p]
\centering
\small
\caption{Prospective operationalizations for the six motivating episodes.
Each row is a research design sketch, not a result.}
\label{tab:cases}
\begin{tabularx}{\textwidth}{P{0.15\textwidth}P{0.22\textwidth}YY}
\toprule
Episode & Candidate clock fixed before outcome analysis &
State and observable & Evidence against the proposed account \\
\midrule
Black Lives Matter, especially 2020 &
Weighted cumulative count of independently verified local protests and
cross-community exposures; measurement informed by prior social-media protest
work \citep{freelon2018}. &
Community orientation, activation, $r_1$, $r_2$, geographic reach. &
Chronological time predicts held-out mobilization as well as or better than
$\tau$; inferred alignment is an artifact of platform sampling.\\
\addlinespace
Brexit referendum and implementation &
Prespecified sequence of campaign exposures, official votes, court decisions,
and implementation milestones \citep{hobolt2016}. &
Distributions of policy preference and institutional commitment; bipolarity
$P$. &
The clock merely encodes the known outcome, or preference changes are better
explained by stable covariates and shocks outside the event set.\\
\addlinespace
Arab uprisings, 2010--2012 &
Within-country protest events and verified cross-border media exposures, with
country-specific measurement models \citep{tufekci2012}. &
Participation, network reach, institutional response, cross-country
heterogeneity. &
One shared clock fails measurement invariance; country outcomes require
mechanisms absent from the phase model.\\
\addlinespace
COVID-19 communication and policy &
Separate epidemiological, policy, and information-event clocks rather than one
post-hoc composite \citep{cinelli2020}. &
Policy attitudes, protective behavior, vaccine attitudes, $r_1$ and $P$. &
Apparent synchronization disappears after accounting for disease incidence,
policy coercion, bots, or changing sample composition.\\
\addlinespace
South Korean candlelight protests, 2016--2017 &
Weekly protest events plus a separately coded sequence of legislative and
judicial steps \citep{kang2019}. &
Participation, impeachment support, institutional state, recovery after the
decision. &
The outcome follows institutional procedure without a detectable coupling
transition, or alternative event codings yield incompatible clocks.\\
\addlinespace
Metaverse attention, 2021--2023 &
Product releases, investment announcements, active-device adoption, and
developer activity measured separately \citep{dwivedi2022}. &
Attention activation $\bar a$, adoption, expectation orientation, response
capacity. &
Attention and implementation are not dynamically coupled, or a conventional
hype/adoption model predicts better with fewer degrees of freedom.\\
\bottomrule
\end{tabularx}
\end{table}

The table also shows why causal potential should not be treated as a single
latent energy. Protest events, institutional acts, disease incidence, media
exposure, investment, and adoption have different measurement processes. A
model may contain several clocks or a vector of event counts; combining them
requires a theory and validation data.

\section{Falsifiable Hypotheses and Empirical Protocol}
\label{sec:empirical}

\subsection{Hypotheses}

The operational framework yields the following hypotheses. Each is stated
with a failure condition.

\begin{description}
  \item[H1: cross-context clock invariance.] For processes governed by the same
  mechanism but exposed to different event rates, trajectories expressed in a
  preregistered $\tau$ are more invariant than trajectories in $t$.
  Failure occurs if held-out alignment or likelihood does not improve relative
  to chronological-time and flexible-time baselines.

  \item[H2: conditional coupling onset.] After estimating heterogeneity,
  diffusion, and network scale, sustained $r_1$ should emerge near the
  conditional threshold appropriate to the fitted model. Failure occurs if
  no parameter-stable onset appears or if null networks reproduce it.

  \item[H3: distinct consensus and bipolarity.] Episodes with two opposed
  clusters should show high $r_2$, low $r_1$, and high $P$, whereas consensus
  should show high $r_1$ and low $P$. Failure occurs if these observables do
  not distinguish independently labeled states.

  \item[H4: anchoring and intervention timing.] Controlling for exposure and
  selection, stronger anchoring $q_i$ should predict longer exit times after a
  standardized counter-message or institutional change. Failure occurs if
  estimated exit times are unrelated to independently measured anchoring.

  \item[H5: consensus is insufficient for depletion.] Low response capacity
  $C_h$, not high $r_1$ alone, should predict failure to adapt to novel input.
  Evidence that aligned groups adapt as readily as dispersed groups would
  reject the stronger consensus-causes-depletion claim while remaining
  consistent with the narrower model.
\end{description}

\subsection{Measurement and estimation workflow}

A credible test should be prospective or pseudo-prospective:
\begin{enumerate}
  \item Define the units, orientation scale, activation measure, network
  boundary, event types, and event weights without consulting the target
  outcome period.
  \item Estimate or cross-fit $\lambda(t\mid\mathcal H_t)$ on a training window.
  Record uncertainty in the resulting $\tau$ rather than treating it as exact.
  \item Establish that orientation measurements are comparable across groups
  and time. Compute $r_1$, $r_2$, $P$, $\bar a$, and, where possible, the
  perturbational response $C_h$.
  \item Fit at least three models: $M_0$, the same dynamics in chronological
  time; $M_1$, TCU dynamics under the fixed causal clock; and $M_2$, a flexible
  time-warp baseline with equal or penalized complexity.
  \item Compare preregistered held-out log likelihood, calibration, forecast
  error, and cross-context parameter stability. Use temporal block
  cross-validation rather than random shuffling.
  \item Repeat the analysis on degree-preserving, timestamp-shuffled, and
  outcome-shuffled nulls. Report sensitivity to event weights, missing nodes,
  bots, and platform sampling.
\end{enumerate}

The decisive comparison is not whether $M_1$ fits, but whether it predicts
better than $M_0$ without gaining the arbitrary flexibility of $M_2$.

\subsection{Identification of influence}

Interaction data commonly confound influence with homophily and shared
exposure. A descriptive fit of \cref{eq:phase} may estimate predictive
coupling, but it does not identify $W_{ij}$ as a causal effect. Causal
interpretation requires a design: randomized message exposure, exogenous
platform or policy changes, instrumental variables with defensible exclusion
restrictions, or longitudinal structural assumptions accompanied by
sensitivity analysis. If those conditions are absent, the paper should use
\emph{association network} and reserve \emph{causal} for the philosophical
level and the prespecified event clock.

\subsection{Ethical boundary}

Models of intervention timing and collective alignment could be used for
manipulation. Empirical work should favor aggregate reporting, data
minimization, informed consent where individual intervention occurs, and
independent ethical review. The framework does not justify steering public
opinion toward a preferred direction. Its scientifically legitimate use is
to describe, forecast with uncertainty, and test mechanisms.

\section{Relation to Existing Theory}
\label{sec:related}

TCU inherits its process orientation from accounts that treat becoming as
fundamental \citep{heidegger1962,whitehead1978}. It differs from physical
accounts of the temporal arrow \citep{reichenbach1956,prigogine1980,price1996}
by making no claim to derive macroscopic irreversibility from microphysics.
The causal-progress clock is closer to event rescaling than to a new physical
dimension \citep{brown2002}.

Mathematically, \cref{eq:phase} belongs to the family of coupled-oscillator
models \citep{kuramoto1984,acebron2005,arenas2008}; the order parameter and
conditional threshold are established results used here as a disciplined
substrate. The proposed contribution is their integration with a separately
measured event clock, an activation layer, anchoring, and explicit
philosophical semantics. Threshold and cascade models
\citep{granovetter1978,watts2002,centola2007} may outperform oscillators when
adoption is discrete or reinforcement is complex. Such comparisons are part
of the falsification program, not competing views to be dismissed.

Social-dynamics work has long shown that averaging, network structure, and
selective exposure can produce consensus, fragmentation, or persistent
disagreement
\citep{degroot1974,friedkin1990,castellano2009,delvicario2016}. TCU does not
claim priority for these mechanisms. Its value, if borne out, would lie in
showing that a constrained causal-progress coordinate improves transfer across
episodes with different calendar speeds.

\section{Limitations and Boundary Conditions}
\label{sec:limits}

First, the ontological identity proposed by TCU is underdetermined by the
network model. The same equations can be interpreted without accepting the
philosophical thesis. Empirical success would support the operational clock,
not prove metaphysical identity.

Second, $\tau$ is scalar. Real systems may contain multiple asynchronous
processes---legal, epidemiological, economic, and informational---that cannot
be compressed without loss into one ordering. Partial orders or vector clocks
may be more appropriate.

Third, circular phase is suitable for orientations with periodic or
directional structure but not for every belief. Discrete threshold models,
bounded-confidence dynamics, or multidimensional latent spaces may be better
for particular domains.

Fourth, the mean-field critical value in \cref{eq:kc} is conditional. It
cannot be transferred to empirical networks by relabeling $D$ as noise and
$W$ as causality. Every application must state its normalization and estimate
uncertainty.

Fifth, historical analogies are vulnerable to hindsight, selection, and
survivorship bias. The six cases in \cref{tab:cases} have not been fitted and
do not validate the model. A future study should use event-level data and
precommitted alternatives.

Finally, terms such as condensation, freezing, and phase transition are useful
only when operational definitions accompany them. They do not imply quantum
superposition, wave-function collapse, or the quantum Zeno effect in human
cognition.

\section{Conclusion}
\label{sec:conclusion}

Temporal--causal unity begins from a strong intuition: time is experienced not
as an empty axis but through the directed realization of change. This paper
turns that intuition into a constrained scientific proposal by separating
interpretation from measurement and measurement from dynamics. The
causal-progress coordinate $\tau$ is meaningful only when its event intensity
is specified independently; the network model is meaningful only when its
states, scales, and influence design are explicit; historical episodes are
informative only when they generate risky predictions rather than retrospective
labels.

Under those constraints, TCU offers a coherent program for studying collective
tempo, alignment, polarization, anchoring, and adaptive capacity. Its next
step is not a broader metaphysical claim but a narrower empirical test:
preregister one clock, compare it with chronological and flexible-time
baselines, and evaluate prediction on held-out processes. Failure of causal
time to improve invariance would count against the operational framework.
Success would not replace established physics, but it would show that the
amount of causally relevant activity is a useful clock for collective
dynamics.

\section*{Acknowledgments and tool-use disclosure}

The conceptual framework originated with the human author. OpenAI Codex was
used to assist with English drafting, mathematical consistency checks, LaTeX
preparation, and simulation-code scaffolding. The human author must verify the
claims, references, code, and final wording and assumes responsibility for the
submitted work. Generative software is not an author.

\appendix

\section{Derivation and Invariance Notes}
\label{app:derivation}

\subsection{Mean-field fixed points}

Writing \cref{eq:meanfield} as
$\dot r_1=r_1[A-(K/2)r_1^2]$ with
$A=K/2-(\Delta+D)$ gives the incoherent fixed point $r_1=0$. If $A>0$,
the nonzero fixed point satisfies
\[
  r_1^2=\frac{2A}{K}=1-\frac{2(\Delta+D)}{K}.
\]
At $K=K_c$, the branch emerges continuously. This is a supercritical onset
for the stated special case; discontinuous or hysteretic transitions require
other distributions, coupling rules, adaptation, inertia, or correlations.

\subsection{Rotation invariance}

Under a global change of angular origin
$\theta_i\mapsto\theta_i+\phi$, one has
$Z_m\mapsto e^{\mathrm{i}m\phi}Z_m$, so $r_m$ and $P$ are unchanged. In
contrast, $\langle|\theta_i|\rangle$ changes with $\phi$ and is therefore not
an intrinsic measure of extremity on the circle.

\subsection{Clock scaling}

If all event weights are multiplied by $c>0$, then $\tau'=c\tau$. Parameters
with inverse-$\tau$ units transform as
$\nu_i'=\nu_i/c$, $K'=K/c$, $q_i'=q_i/c$, and $D_i'=D_i/c$.
Predictions are unchanged only if this scale convention is propagated. A
paper must therefore report the event-weight normalization before comparing
parameters across datasets.

\section{Traceability to the Motivating Hypotheses}
\label{app:traceability}

\begin{table}[h]
\centering
\small
\caption{Status of the twelve motivating hypotheses after operationalization.}
\label{tab:trace}
\begin{tabularx}{\textwidth}{P{0.09\textwidth}P{0.34\textwidth}Y}
\toprule
Original & Motivating statement & Status in the present framework \\
\midrule
H1 & Time and causation are two aspects of one essence. &
Retained as TCU-1, explicitly interpretive rather than proven.\\
H2 & Space is a synchronic slice of temporal--causal unfolding. &
Retained only as a descriptive state-slice interpretation; no elimination of
physical space.\\
H3 & Causation is the tangent vector of temporal unfolding. &
Replaced by the measurable event intensity and directed state drift; not a
differential-geometric claim.\\
H4 & Cognition is a complex amplitude on a time--cause plane. &
Recast as the classical state $\psi_i=a_i e^{\mathrm{i}\theta_i}$; no quantum
meaning.\\
H5 & Temporal--causal unfolding obeys a dynamical equation. &
Retained as the explicitly scoped stochastic model
\eqref{eq:phase}--\eqref{eq:amplitude}.\\
H6 & Synchrony $r$ is the system order parameter. &
Refined to $r_1$, $r_2$, $P$, activation, and response capacity.\\
H7 & Cognition is subjectivized temporal--causal unfolding. &
Retained as interpretation; empirical units require validated cognitive
measures.\\
H8 & Information is a medium of causal potential. &
Replaced by measured input $\boldsymbol x$, event intensity $\lambda$, and
identified or descriptive influence.\\
H9 & An information cocoon is a fixed unfolding direction. &
Recast as persistent concentration plus selective exposure and low response
capacity.\\
H10 & Society is an interwoven temporal--causal network. &
Operationalized by $W$ and timestamped interactions, without assuming every
edge is causal.\\
H11 & Consensus is synchronization. &
Partially retained; harmonics distinguish consensus from bipolar order.\\
H12 & Social change is a phase transition. &
Retained as a conditional model hypothesis, not a universal description.\\
\bottomrule
\end{tabularx}
\end{table}

\bibliographystyle{plainnat}
\bibliography{references}

\end{document}